\documentclass[12pt]{article}
\usepackage{geometry}
\usepackage{url}
\usepackage{latexsym}
\usepackage[centertags]{amsmath}
\usepackage{amssymb}
\usepackage{amsthm}
\usepackage{enumerate}
\usepackage{graphicx}
\usepackage{amsfonts,amssymb}
\usepackage{amsmath}
\usepackage[all]{xy}   
\usepackage{tikz-cd}

\usepackage{tikz}

\usepackage{tikz-cd} 
\usetikzlibrary{shapes.misc,arrows,decorations.markings}
\usetikzlibrary{matrix}
 
\newtheorem{Theorem}{Theorem}

\newtheorem{Proposition}{Proposition}  
\newtheorem{Lemma}{Lemma}

\newcommand {\cC}{\mbox{${\mathcal C}$}}

\newcommand{\R}{\mathbb{R}}

\newcommand{\Cinf}{C^\infty }

\newcommand{\Mstar}{\mbox{$M_*$}}

\DeclareMathOperator{\scl}{sc}
\DeclareMathOperator{\topol}{top}

\begin{document}

\title{Through the Singularity}

\author{Michael Heller, Tomasz Miller \\ \normalsize Copernicus Center for Interdisciplinary Studies, Jagiellonian University \\ 
\normalsize Szczepa\'{n}ska 1/5, 31-011 Cracow, Poland \\ 
Wies{\l}aw Sasin\\ \normalsize Institute of Mathematics and Cryptology, Military University of Technology\\
\normalsize Kaliskiego 2, 00-908 Warsaw, Poland \\[12pt]}

\date{\today}
\maketitle

\begin{abstract}
In this work, we propose a dangerous journey -- a journey through the strong  singularity from one universe to another or from inside of a black hole to its 'inverse' as a white hole. Such singularities are hidden in the Friedman and Schwarzschild solutions; we call them malicious singularities. The journey is made possible owing to two generalizations. The first generalization consists in considering spaces with differential structures on them (the so-called ringed spaces) rather than the usual manifolds. This entails a generalization of the concept of smoothness, which allows us to think about a smooth passage through the singularity. The second generalization is related to the concept of curve. We show that if a kind of singularity is implanted in the set of curve's parameters, along with an appropriate topology, in such a way that the structure of the set of parameters corresponds to the structure of the singular space-time, the curve can smoothly -- in a generalized sense -- pass through the singularity. 
\end{abstract}

\section{Introduction}
The rapidly developing physics of black holes, both theoretically and observationally, brings the singularity problem into sharp focus once again. In the classical period of studying singularities and formulating theorems about their existence (the 1970s and 1980s of the previous century), it was still possible to neutralize the problem by eliminating the beginning of the universe with the help of various tricks. Today, when observations clearly suggest that there are black holes at the centers of most galaxies, the problem must be faced in all its severity. Until we have a final quantum theory of gravity, classical methods should be fully exploited. They may not only be valuable in themselves, but also point the way to the future quantum gravity theory.

In this work, we want to do the impossible: not only reach the singularity, but also go through it to a new universe, if one exists 'on the other side'. This is not an easy undertaking. All the time, we have in mind strong singularities, such as the ones hidden in the Friedman or Schwarzschild solutions; in the following we will call them malicious singularities. So far, such singularities have been understood as points of a space-time boundary (not belonging to space-time, so its points are not points in the usual sense of the word) at which the histories of observers, photons or other particles terminate. If so, then these histories -- by definition -- cannot go through the singularity. However, in mathematics, sometimes a small generalization is enough to overcome another conceptual barrier.

In this work, two such generalizations turn out to be crucial. The first concerns the concept of smoothness. We use this concept, as it is understood in the theory of Sikorski's differential spaces \cite{Sikorski} and the theory of Grothendieck's sheaf spaces \cite{Grothendieck}. The essence of this approach is to consider, instead of the usual differential manifolds, a more general class of spaces with differential structures defined on them. These so-called \emph{ringed spaces} will be recalled in Section \ref{sec::2}. This method is employed to analyze malicious singularities, regarded as points of the $b$-boundary of space-time. A short review of this issue can be found in Section \ref{sec::3}. The second generalization is related to the concept of curve. A smooth curve (in the traditional sense) has no chance of passing through a malicious singularity, but if a kind of singularity is implanted in its set of parameters (along with an appropriate topology), so as the structure of the set of parameters corresponds to the structure of the singular space-time, the curve can smoothly -- in a generalized sense -- pass through the singularity. This is shown in Section \ref{sec::4}.

\section{The Concept of Smoothness}
\label{sec::2}
Let us recall the standard concept of smoothness as we encounter it in every analysis course. Let $M$ be an $n$-dimensional differential manifold, and $\mathcal{A} = \{(U_{\alpha}, \phi_{\alpha })\}_{\alpha \in I}$ an atlas on $M$. A map $f: M \to \R$ is said to be a smooth function on $M$ if, for every $p \in M$, there exixts a chart $(U, \phi) \in \mathcal{A}$ such that $p \in U$ and $f \circ \phi^{-1}: \phi(U) \to \R $ is smooth in the usual sense. The smoothness condition for the transition maps between charts in the atlas guarantees that if $f$ is smooth in one chart (in a neighbourhood of $p$), it is smooth in any other chart.

Let $N$ be another $n$-dimensional manifold, and let us consider a map $F: M \to N$. The map $F$ is said to be smooth if, for every $p \in M$, there exists a chart $(U, \phi)$ on $M$ with $p \in U$, and a chart $(V, \psi )$ on $N$ with $F(p) \in V$, such that $F(U) \subset V$ and $\psi \circ F \circ \phi^{-1}: \phi (U) \to\psi(V)$ is a smooth function on $\R^n$.

Any smooth mapping $F: M \to N$ between manifolds induces linear maps between their tangent spaces $F_{*p}: T_pM \to T_{F(p)}N$. With the help of this map and its dual, vector fields, differential forms and other local geometric data can be transported between manifolds.

To deal with singularities we need a generalization of this smoothness concept. We find it in Sikorski's theory of differential spaces. His idea was to collect all functions that we would like to consider 'smooth' in a family defined on the space $M$ called its differential structure. Of course, for this concept to be reasonable, the family must satisfy suitable conditions. This idea has been implemented in the following way (for a short introduction to the theory of differential spaces see \cite{GruszczHel93,Structured}).  

Let $(M, \topol M)$ be a topological space, and $C$ a family of functions defined on $M$ (with values in any set). A function $f$, defined on an open subset $V \subset M$, is said to be a local $C$-function on $V$ if for every $p \in V$ there exists its open neighbourhood $U \subset V$ and a function $g \in C$ such that $f|U = g|U$. The set of all local $C$-functions on $V \subset M$ is denoted by $C_V$.

From the definition it follows that $C|V \subset C_V$, and if $V = M$ then $C \subset C_M$. If $C = C_M$ (every local $C$-function on $M$ belongs to $C$), we say that the family $C$ is closed with respect to localization. This is the property we would require for smooth functions. Another such property is the following.

The non-empty family $C$ of real valued functions defined on a set $M$ is said to be closed with respect to composition with smooth functions provided the following condition is satisfied: if $\omega \in \Cinf (\R^n)$ and $f_1, \ldots , f_n \in C$, then $\omega (f_1, \ldots , f_n) \in C$. The family $C$ satisfying this condition is denoted by $\scl C$.

A family $C$ of real functions on $M$, which is closed with respect to localization and closed with respect to composition with smooth functions is called the differential structure on $M$, and the pair $(M, C)$ the differential space. All functions belonging to $C$ are smooth \textit{ex definitione}.

A mapping between differential spaces is smooth if it preserves the smoothness understood in the above way, i.e. if $(M, C)$ and $(N, D)$ are differential spaces, then the map $f: M \to N$ is smooth if, for every function $g \in D$, $g \circ f \in C$.

The above understanding of smoothness works well in many cases, but is still ineffective when applied to stronger type singularities. What we need is a `localization' of this concept. This is achieved in the following way.

Let $(M, \topol M)$ be a topological space. The sheaf $\cC $ of real continuous functions on $(M, \topol M)$ is said to be a differential structure if, for any $U \in \topol M$, one has $\scl \cC (U) = \cC (U)$, i.e. every $\cC (U)$ is closed with respect to composition with smooth functions. The pair $(M, \cC )$ is called the structured or sheaf space (for a theory of structured spaces see \cite{Structured}).

Since the sheaf definition implies that each  $\cC(U)$ is closed with respect to localization, every $\cC (U)$ is a differential structure on $U$ (in the sense of Sikorski). 

Let $(M, \cC )$ and $(N, \mathcal{D})$ be two structured spaces. A mapping 
\begin{align*}
F:(M, \cC) \to (N, \mathcal{D})
\end{align*}
is said to be smooth if $F^*(\mathcal{D}) \subset \cC $, i.e. for every $U \in \topol N$ and $g \in \mathcal{D}(U)$ one has $F^* g := g \circ F \in \cC(F^{-1}(U))$.

Structured spaces as objects and smooth mappings between them as morphisms form a category of structured spaces. We will use the concept of smoothness involved in this category in our further considerations regarding going through a malicious singularity.

\section{Malicious Singularity}
\label{sec::3}
The problem is that the singularity cannot belong to space-time, because the metric and even manifold structures of space-time break down in it but, on the other hand, it must be defined using theoretical tools available within space-time, otherwise, it would not even have a quasi-operational significance. Moreover, the definition of a singularity should also be `theoretically operational', i.e. it should be possible to prove statements about its occurrence and properties based on it.

As a result of long discussions (described in detail in \cite{TCE}), it turned out that these criteria are met by the definition of a singularity as a geodesic incompleteness of space-time. An inextendible space-time\footnote{Space-time $(M', g')$ is an extension of space-time $(M, g)$ if there is an isometric embedding $\rho: M \to M'$. Space-time is inextendible if there is no such embedding with $\rho(M) \neq M'$.} is geodesically complete if all (timelike and null) geodesics in it can be continued to arbitrarily large values of an affine parameter (in both directions). If this is not the case, a given space-time is said to be geodesically incomplete. An incomplete geodesic represents the history of an object that `runs into a premature end' \cite[p. 139]{TCE}, and may be regarded as ending at a singularity.  

Singularities, understood in this way, can be organized into a kind of  boundary of space-time, the so-called $g$-boundary. The point of the $g$-boundary of space-time is defined by an equivalence class of incomplete geodesics that meet the condition allowing us to regard them as ending at the same point.

As we can see, geodesical incompleteness is a sufficient condition for the existence of a singularity rather than its full fledged definition, but it well served its purpose since it was possible, with the help of it, to prove several theorems on the  existence of singularities \cite{HawkingEllis}. And mostly for this reason, the $g$-boundary construction has  outdistanced other similar constructions such as causal boundary \cite{KronheimerPenrose}, $p$-boundary  \cite{Dodson1} or essential boundary \cite{Clarke1}, but it was not without problems itself. One of them was that geodesic  (in)completeness applies, as the name suggests, only to  geodesics, while also timelike curves can, in principle, 'break off' at a singularity. It was Geroch who constructed a geodesically complete space-time containing an inextendible timelike curve of finite length and bounded acceleration \cite{Geroch1968}. The latter property is essential since `an observer with a suitable rocketship and a finite amount of fuel could traverse this curve' \cite[p. 258]{HawkingEllis}.

This shortcoming was supposed to be remedied by the construction proposed by Schmidt \cite{Schmidt71}. Since it plays an important role in our considerations, let us present it briefly (for a detailed account \cite{Clarke,Dodson2,Schmidt71}). A connection on space-time $M$ induces a Riemann metric on a connected component of the fibre bundle $O(M)$ of orthonormal frames over $M$. It does it in such a way as to make horizontal and vertical subspaces orthogonal. This leads to the existence of the Riemann metric $G$ on $O(M)$. This metric is not unique but the subsequent steps of the construction do not depend of the choice of the one of admissible metrics. With the help  of $G$ we define the distance function and construct the Cauchy completion $\overline{O(M)}$ of $O(M)$. We extend, by continuity, the action of the structural group $O(3, 1)$ of the bundle from $O(M)$ to $\overline{O(M)}$ and we form the quotient space $\overline{O(M)}/O(3, 1)$ to define
\begin{align*}
\partial_b M := \Mstar - M = \pi(\overline{O(M)}) - \pi (O(M))
\end{align*}
where $\pi: \overline{O(M)} \to \overline{O(M)}/O(3, 1)$ is the canonical projection, and $M_* = M \cup \partial_b M$ the space-time with its $b$-boundary.

The elements of the $b$-boundary of space-time are now the equivalence classes (of a suitable equivalence relation) of not necessarily causal geodesics, but also the equivalence classes of other curves. If, however, $b$-boundary is restricted to causal geodesics, it changes into $g$-boundary of space-time. So it seemed that the problem of a 'working definition' of singularity had been solved. It was an unpleasant surprise when Bosshard \cite{Bosshard} and Johnson \cite{Johnson} independently discovered the following pathologies in which $b$-boundary are involved.  If the singularity in the closed Friedman solution or Schwarzschild solution of Einstein field equations is represented as a point $p$ in the $b$-boundary, $p\in \partial_bM$, then the fibre $\pi^{-1}(p)$ of the fibre bundle of linear frames over $p$ degenerates to a single point. This is particularly frustrating with respect to Friedman's closed solution, because in that case the beginning and end of the universe would be the same, and only, point of the boundary. Moreover, this $b$-boundary point turns out not to be Hausdorff separated from the rest of space-time. In fact, topology of this configuration is highly degenerate: the only open neighbourhood of the singularity is the whole of space-time. These conclusions seemed disastrous for the $b$-boundary construction: what was supposed to be a successful definition of singularity turned out to be a source of new difficulties.

The situation was clarified in \cite{SH94,Structured}. It turned out to be essential to employ the concept of smoothness as it is integrated into the construction of space-time as a structured space. The following theorem is crucial.
\begin{Theorem}\label{Theorem1}
Let $(M, \cC )$ be a structured space with the topology $\tau $, the weakest topology in which functions of $\cC $ are continuous, and $(M_*, \tau )$ a topological space such that $M_* = M \cup \{*\}$, $* \notin M$. And let the following conditions be satisfied
\begin{itemize}
\item 
$\tau|M = \tau_{\scriptsize\cC(M)}$,
\item 
$* \in U \in \tau \Rightarrow U = M$.
\end{itemize}
Then on the topological space $(M, \tau )$ there exists exactly one differential structure $\cC_*$ such that $\cC_*(M) = \cC(M)$ and $\cC_*(M_*) = \R $.
\end{Theorem}
\begin{proof} See \cite{Structured}.
\end{proof} 


This theorem says that if we are inside space-time $M$, i.e. in any open set $U \in \tau $, we will not find any pathologies, everything will happen as in a regular (non-singular) space-time. However, if we want to prolong the differential structure $\cC $ from $M$ to $M_*$, this can only be done in a trivial way, that is only constant functions can be prolonged, and we obtain $\cC(M_*) = \R$. Therefore, the whole of $M_*$ collapses to a single point. 

We also have the following theorem.
\begin{Theorem} \label{Theorem2}
Let $p \in \partial_bM$. If the fibre $\pi^{-1}(p)$ in the fibre bundle of orthonormal frames over $M_*$ degenerates to a single point, then only constant functions can be prolonged to $M_*$, i.e. $\cC(M_*) = \R$.
\end{Theorem}
\begin{proof} See \cite{Structured}\footnote{The set of all linear frames over $M_*$, if there exist in it degenerate fibres, is not a fibre bundle, in the usual sense, but it is a fibre bundle in the category of structured spaces.}.
\end{proof}

Singularites to which Theorems \ref{Theorem1} and \ref{Theorem2} apply we have in anticipation called malicious singularities.

Theorems \ref{Theorem1} and \ref{Theorem2} at least partially explain 'pathological' situations analyzed by Bosshard \cite{Bosshard} and Johnson \cite{Johnson}. These theorems confirm these `pathologies': their differential structures indeed consist of only constant functions, and their topology of only one open set (besides empty set). Therefore, the differential structure of the closed Friedman model encodes the beginning and end of the Friedman universe (and all other its points) as just one point. However, as soon as we do not extend the differential structure to the $b$-boundary (to the singularity), everything happens as in ordinary space-time.

In our previous works, we stopped at these explanations, now we want to pick up this thread and go further. So we ask the question: if we have the apparatus of structural spaces and its possible sharpening at our disposal, is it possible to go (smoothly or not) through the malicious singularity?  This will be the subject of our analysis in the following section. To move forward, we first have to design curves that could do the work. 

\section{Curves in the Singularity}\label{sec::4}
Let $(M_*, \cC_{M_*})$, $M_* = M \cup \{*\}$, be a space-time manifold $M$ with the malicious singularity $* \in \partial_bM$, where $\cC_{M_*}$ is the sheaf of smooth functions on $M_*$. The topology on $M_*$ is
\begin{align*}
\topol M_* = \topol M \cup \{M_*\}
\end{align*}
where $\{ M_*\} $ is the only open subset containing the singularity $*$. The cross-sections of the sheaf $\cC_{M_*}$ are of the form
\begin{align*}
\cC_{M_*}(M_*) = \R, \qquad \cC_{M_*}(U) = \Cinf (U)
\end{align*}
for $U \in \topol M$.

We will analyze curves passing through the singularity. Let us consider two types of smooth curves (smooth -- in the sense of structured spaces):
\begin{itemize}
\item
$\gamma: I \to M_*$ where $I = \langle 0, 1\rangle$ is endowed with the standard topology, $\gamma|\langle 0, 1) \subset M$ and $\gamma (1)  = *$,
\item
$\gamma_*: I_* \to M_* $ where $I_* = \langle 0, 1\rangle$ is endowed with the topology $\topol \langle0,1) \cup \{I_*\}$ (mimicking the topology of $M_*$), along with $\gamma_*|\langle 0, 1) \subset M$ and $\gamma_*(1)  = *$ as above.
\end{itemize}

Let us notice that the injection $\imath : I \to I_*$, defined simply as the identity mapping $\imath(t) = t, \, t \in \langle0,1\rangle$, is continuous.

The differential structure on $I_*$ is defined as a sheaf on $I_*$, $\cC_{I_*} = \R, \, \cC_{I_*}(U) = \Cinf (U)$ for $U \in \topol \langle 0, 1 )$. Our curves can be thus regarded as morphisms
\begin{align*}
& \gamma : (I, \Cinf ) \to (M_*, \cC_{M_*}),
\\
& \gamma_* : (I_*, \cC_{I_*} ) \to (M_*, \cC_{M_*}).
\end{align*}

It can be easily seen that the injection $\imath$ is a smooth mapping of parameter spaces of structured spaces $(I, \Cinf )$ and $(I_*, \cC_{I_*} )$. This allows us  to define the transformation of the curve $\gamma_* $ into the curve $\gamma $.
\begin{align*}
\gamma = \gamma_* \circ \imath .
\end{align*}
(see the commutative diagram below).
\begin{center}
\begin{tikzcd}
I \arrow[d, "\imath " description] \arrow[r, "\gamma" description] & M_* \\
I_* \arrow[ru, "\gamma_*" description]                        &    
\end{tikzcd}
\end{center}

The following lemmas present the smoothness properties of curves $\gamma $ and $\gamma_*$.
\begin{Lemma} \label{Lemma1}
Curve $\gamma : I \to M_*$ is smooth if and only if curve $\gamma|\langle 0, 1): \langle 0, 1) \to M$  is smooth.
\end{Lemma}
\begin{proof} 
Recall that the smoothness of $\gamma$ means that $\gamma^*(\cC_{M_*}(U)) \subset \Cinf(\gamma^{-1}(U))$ for any $U \in \topol M$ and $\gamma^*(\cC_{M_*}(M_*)) \subset \Cinf(\gamma^{-1}(M_*))$. But since $\cC_{M_*}(U) = \Cinf (U)$, the former condition is equivalent to the smoothness of $\gamma|\langle 0, 1)$, whereas the latter condition boils down to $\gamma^*(\R) \subset \Cinf(\langle 0, 1))$ (with the elements of $\R$ interpreted as the constant functions), what is trivially true.
\end{proof}

\begin{Lemma} \label{Lemma2}
Curve $\gamma_*: I_* \to M_*$ is smooth if and only if curve $\gamma : I \to M_*$ is smooth.
\end{Lemma}
\begin{proof} The smoothness of $\gamma_*$ means that $(\gamma_*)^*(\cC_{M_*}(U)) \subset \cC_{I_*}(\gamma_*^{-1}(U))$ for any $U \in \topol M$ and $(\gamma_*)^*(\cC_{M_*}(M_*)) \subset \cC_{I_*}(\gamma_*^{-1}(M_*))$. By defition of the respective sheaves, the latter condition boils down to $(\gamma_*)^*(\R) \subset \R$, what is trivially true (even as an equality). On the other hand, the former condition reduces to $\gamma^*(\Cinf(U)) \subset \Cinf(\gamma^{-1}(U))$, where we have replaced $\gamma_*$ with $\gamma$ because they are equal as maps ($\gamma_*(t) = \gamma(t)$ for all $t \in \langle 0, 1\rangle$). But as we have seen in the proof of Lemma \ref{Lemma1}, this is already equivalent to the smoothness of $\gamma$.

\end{proof}

As we can see, both curves, $\gamma $ and $\gamma_*$ are smooth in the category of structured spaces, even though the topologies of their parameter spaces are different, $\topol I_* \subset \topol I$.

Let us now consider a curve $\gamma $ hitting the malicious singularity, such as the one likely hiding in the Big Bang or beyond the horizon of a black hole; alternatively, a curve passing through the singularity from an expanding phase of the universe to its expanding phase or, if you prefer, entering a black hole singularity, and then emerging from the `other side of it' as of a white hole. For brevity, we will talk about the transition from the universe $M_1$, through the singularity $*$, to the universe $M_2$, where $M_1 \cap M_2 = \emptyset$.

We thus have the following situation: $M_* = M_1 \cup \{*\} \cup M_2$ ($M_*$ will also be denoted by $M_1 \star M_2$) with the topology
\begin{align*}
\topol M_* = \topol (M_1 \cup M_2) \cup \{M_*\}
\end{align*}
where $M_*$ is the only open set containing the singularity.

For any smooth curves $\gamma_1: (0,1) \rightarrow M_1$ and $\gamma_2 : (1,2) \rightarrow M_2$ let us now define the curve $\gamma_1 \star \gamma_2: (0,2) \rightarrow M_*$ ``passing through'' the singularity
\begin{align*}
(\gamma_1 \star \gamma_2)(t) =
\left\{
\begin{array}{cc}
\gamma_1(t) \, \in M_1 & \mbox{if $t \in (0, 1)$}\\
*  & \mbox{if $t = 1$}\\
\gamma_2(t) \, \in M_2 & \mbox{if $t \in (1,2)$}
\end{array}
\right. 
\end{align*}
where the set of parameters, denoted $I_* = (0,2)$ in what follows, is endowed with the topology
\begin{align*}
\topol I_* = \topol ((0, 1) \cup (1, 2)) \cup \{(0, 2)\}.
\end{align*}

As a kind of formal summary of our considerations, we can formulate the following proposition
\begin{Proposition}
The curve $\gamma_1 \star \gamma_2: I_* \to M_*$ is continuous and smooth if and only if $\gamma_1$ and $\gamma_2$ are continuous and smooth.
\end{Proposition}
\begin{proof}
The continuity part is obvious from the very definition of the topologies involved. As for the smoothness part, the reasoning goes along the same lines as in the proofs of Lemmas \ref{Lemma1} \& \ref{Lemma2} with some simple modifications, such as replacing $I$ and $M$ with $(0,1) \cup (1,2)$ and $M_1 \cup M_2$, respectively.
\end{proof}

This method of passing through the singularity opens up yet another problem. The Noether theorem shows that breaking the translational time symmetry leads to violating the principle of energy conservation (see, e.g. \cite{Hanca,Kosmann}). Introducing a singularity to the reparamerization of the curve $\gamma_1 \star \gamma_2$ and synchronizing it with the cosmological singularity creates just such a situation. Although the curve $\gamma_1 \star \gamma_2$ passes smoothly through the singularity (in the sense of the theory of differential spaces), at the ``transition'' moment anything can happen.

Even more importantly, let us notice that the above construction is very general. In fact, in this way one can ``smoothly join'' even the curves which do not approach the $b$-boundary (in particular, nothing in the above construction prohibits $\gamma_1$ and $\gamma_2$ to be \emph{constant} curves!). It is thus worth investigating whether the above construction can be made sharper in the case of curves that \emph{do} approach the singularity by explicitly involving the $b$-boundary definition.


\begin{thebibliography}{ccc}
\bibitem{Bosshard}
Bosshard, B., On the $b$-boundary of the closed Friedmann model, Commun. Math. Phys. 46, 263--966 (1976).
\bibitem{Clarke1}
Clarke, C.J.S., Boundary definitions, General Rel. Gravit. 10(12), 977--980 (1979).
\bibitem{Clarke}
Clarke, C.J.S., The Analysis of Space-Time Singularities. Cambridge University Press, Cambridge (1993).
\bibitem{Dodson1}
Dodson,  C.T.J., A new bundle completion for parallelizable space-times, General Rel. Gravit. 10, 969--979 (1979).
\bibitem{Dodson2}
Dodson, C.T.J., Space-time edge geometry, Int. J. Theor. Phys. 17, 389--504 (1978).
\bibitem{Geroch1968}
Geroch, R.P., What is a singularity in general relativity? Ann. Phys. 48(3), 526--540 (1968).
\bibitem{Grothendieck}
Grothendieck, A., A General Theory of Fibre Spaces with Structure Sheaf, Natural Science Foundation Research Project on Geometry of Function Space, Report No 4, University of Kansas, Lawrence, 1958;
\url{https://webusers.imj-prg.fr/~leila.schneps/grothendieckcircle/Kansasnotes.pdf}
\bibitem{GruszczHel93}
Gruszczak, J., Heller, M., Differential structure of space-sime and its prolongations to singular boundaries, Int. J. Theor. Phys. 32, 625--648 (1993).
\bibitem{Hanca}
Hanca, J., Tulejab, S., Hancova, M., Symmetries and conservation laws: Consequences of Noether's theorem, Am. J. Phys. 72, 428--435 (2004).
\bibitem{HawkingEllis}
Hawking, S. W., Ellis, G. F. R., The Large-Scale Structure of Space-Time. Cambridge University Press, Cambridge (1973).
\bibitem{SH94}
Heller, M., Sasin, W., The structure of the $b$-completion of space-time, General Rel. Gravit. 26, 797--811 (1994).
\bibitem{Structured}
Heller, M., Sasin, W., Structured spaces and their application to relativistic physics, J. Math. Phys. 36, 3644--3662 (1995).
\bibitem{Johnson}
Johnson, R.A., The bundle boundary in some special cases, J. Math. Phys. 18, 898--902 (1977).
\bibitem{Kosmann}
Kosmann-Schwarzbach, The Noether theorems: Invariance and conservation laws in the twentieth century, Sources and Studies in the History of Mathematics and Physical Sciences. Springer-Verlag (2010).
\bibitem{KronheimerPenrose}
Kronheimer, E.H., Penrose, R., On the structure of causal spaces, Proc. Camb. Philos. Soc. 63(2), 481--501 (1967).
\bibitem{Schmidt71}
Schmidt, B.G., A new definition of singular points in general relativity, General Rel. Gravit. 1, 269--280 (1971).
\bibitem{Sikorski}
Sikorski, R.: Introduction to Differential Geometry, PWN, Warszawa, 1972 (in Polish).
\bibitem{TCE}
Tipler, F.J., Clarke, C.J.S., Ellis, G.F.: Singularities and horizons, In: Held, A. (ed.) General Relativity and Gravitation. One Hundred Years after the Birth of Albert Einstein, vol. 2, pp. 97--206. Plenum Press, New York -- London (1980).
\end{thebibliography}
\end{document}